\documentclass[11pt]{article}
\usepackage{amssymb,amsmath}
\usepackage{epsfig}
\usepackage{graphicx}
\usepackage{textcomp}
\usepackage{hyperref}
\usepackage{amsthm}
\usepackage{cleveref}
\usepackage{thmtools,thm-restate}
\usepackage{graphicx}

\newtheorem{theorem}{Theorem}
\newtheorem*{theorem*}{Theorem}

\newtheorem{definition}{Definition}

\newtheorem{lemma}{Lemma}

\newtheorem{corollary}{Corollary}
\newtheorem{fact}{Fact}

\def\mod{\mbox{mod}}

\def\log{{\rm log}}

\newcommand{\be}{\begin{eqnarray}}
\newcommand{\ee}{\end{eqnarray}}

\newcommand\ket[1]{{ |{#1} \rangle }}
\newcommand\bra[1]{{ \langle {#1} | }}

\makeatletter
\newcommand\xleftrightarrow[2][]{%
  \ext@arrow 9999{\longleftrightarrowfill@}{#1}{#2}}
\newcommand\longleftrightarrowfill@{%
  \arrowfill@\leftarrow\relbar\rightarrow}
\makeatother

\renewcommand{\epsilon}{\varepsilon}

\begin{document}
\title{polylog-LDPC Capacity Achieving Codes for the Noisy Quantum Erasure Channel}
\author{Seth Lloyd \thanks{Department of Mechanical Engineering, MIT}, 
Peter Shor \thanks{Department of Mathematics, MIT}, and Kevin Thompson\thanks{School of Engineering and Applied Science, Harvard}}
\maketitle

\begin{abstract}
We provide $poly\log$ sparse quantum codes for correcting the erasure channel arbitrarily close to the capacity.  Specifically, we provide $[[n, k, d]]$ quantum stabilizer codes that correct for the erasure channel arbitrarily close to the capacity if the erasure probability is at least $0.33$, and with a generating set $\langle S_1, S_2, \hdots{} S_{n-k} \rangle$ such that $|S_i|\leq \log^{2+\zeta}(n)$ for all $i$ and for any $\zeta > 0$ with high probability.  In this work we show that the result of Delfosse et al. \cite{DZ13} is tight: one can construct capacity approaching codes with weight almost $O(1)$.
\end{abstract}

\section{Introduction}

Graph States \cite{SW00, H04} are a very useful set of stabilizer states that have found many applications in quantum information theory.  They are important components of measurement based computation schemes \cite{RB01}, and quantum error correcting codes \cite{SW00, GKR02, CSSZ09, MDS08, GJ12, H06, He14, G15, KDP11} to name a few.  One reason they are so ubiquitous in quantum computing is that they have an easy to understand entanglement structure \cite{H04}.  In some sense they provide a ``standard form'' for stabilizer states \cite{CSSZ09, RB01}, since any stabilizer state is locally equivalent to a graph state, and any stabilizer quantum code is equivalent to a quantum graph code\cite{GKR02}.

Indeed, given a graph state over a composite quantum system, it is well known how to compute the entanglement entropy between any partition of the subsystems \cite{H04}.  In this paper Hein shows that the entanglement entropy for a two set partition of a graph state depends in a natural way on the rank of the ``cut'' matrix separating the two sets.  This idea was generalized to qudits and explored to give conditions for secret sharing in \cite{He14}.  Our construction leverages exactly this property of graph states to produce quantum codes with interesting properties.

In particular, we supply quantum stabilizer codes that approach the capacity of the quantum erasure channel\cite{CDS97} and are $\log$-LDPC (Low Density Parity Check).  This means that the stabilizer group of the code can be generated by elements of the Pauli group that are only polylogarithmic with the total number of qubits.  We stress that the codes described make little progress toward the so called quantum LDPC conjecture\cite{BH14, TZ14, H16}\footnote{See the second reference and the references it contains for many other LDPC constructions}.  This conjecture posits the existence of locally generated codes with linear minimum distance.  Our codes are indeed generated by stabilizers which are only logarithmic in their length, but we expect the distance of these codes to also only be logarithmic in the number of qubits, by the arguments presented in \cite{KDP11}.   

The codes described in this paper have very poor adversarial distance ($O(\log(n))$), however the adversarial errors that destroy the encoded information are very unlikely.  This feature allows us to correct for the quantum erasure channel at capacity.  It is known that this cannot be achieved by stabilizer codes of constant weight \cite{DZ13}, we demonstrate that it is possible with a stabilizer code with generators that are only $O(\log^{2+\zeta}(n))$ for any $\zeta > 0$.

It is interesting to contrast the quantum case with the classical.  Classical LDPC codes with linear distance have been known for some time \cite{G63}.  Even random classical LDPC codes can be shown to have linear distance generically.  Constructions of Quantum LDPC codes with linear distance are not known, despite much research in this area\cite{BH14, TZ14, H16}.  In both the classical and quantum regime there are bounds demonstrating that the capacity of the erasure channel cannot be achieved with LDPC codes\cite{DZ13}\cite{BKL02}.  For classical codes, there are known constructions\cite{LM97} with `barely' diverging parity check weight that can correct for the erasure channel very near the capacity.  We provide such a construction for the quantum setting.  Note further that Reed-Muller codes have recently been shown to achieve the capacity of both the classical and quantum erasure channels \cite{SCP16, SP15}.\footnote{In the later case, it was shown that CSS codes where the $X$ and $Z$ checks are Reed-Muller codes can correct for the quantum erasure channel.}  To contrast this work with ours, note that Reed-Muller codes have locality which is linear in the number of qubits.

Our proof techniques are along the lines of standard `typicality' arguments in information theory, along with some simple observations about graph states and the rank of random matrices.  We construct a quantum graph code by randomly sampling an Erdos-Renyi graph and constructing the corresponding graph state.  In addition, we randomly sample the parity check matrix for a classical code, and use these to construct a quantum graph code.  We send this randomly sampled code through the erasure channel, and apply a result from\cite{K98}, as well as standard typicality arguments to show our scheme succeeds with high probability.

\section{Notation}

Graphs will be denoted by the pair $(V, E)$ where $V$ is the vertex set of the graph and $E$ is the edge set.  $V$ is some set of the form $\{1, \hdots{} n\}$.  As is standard, the edge set $E$ will be of the form $\{(i, j)\}$ where each $i, j \in V$.  We will study undirected graphs, so it is important to treat $(i, j)$ as unordered. For a vertex $i \in V$ we denote the neighbor set as usual $N(i)=\{j \in V \,| \, (i, j) \in E\}$.  Denote the adjacency matrix $A$.  We will also denote the ``top half'' of the adjacency matrix as $A_{top}$.  If $i\geq j$ then $A_{top, ij}=A_{ij}$, otherwise $A_{top, ij}=0$.

Let $K$ be a subset of the vertices.  We will denote the cut matrix between $K$ and $V \setminus K$ as $A_{cut}$.  $A_{cut}$ is a $|V\setminus K| \times |K|$ matrix such that if $A_{i j}=1$ for $j \in K$ and $i \in  V \setminus K$ then $A_{cut}$ of the corresponding entry is $1$.  Otherwise $A_{cut}$ is zero.

All vectors in $\mathbb{F}_2^n$ will be column vectors by default, although occasionally we will refer to a row vector with parenthesis, and a concatenation of two vectors with parenthesis.

For classical codes, we will use the standard $[n, k, d]$ notation.  We will say some code $C$ is a $[n, k, d]$ code if $C$ is a $k$ dimensional subspace of $\mathbb{F}_2^n$ such that $\min_{x \in C} |x|=d$, where $|\cdot{} |$ denotes the Hamming weight.  We will denote a code $C$ as a $[n, k, d, w]$ code if the code is a $[n, k, d]$ code and its linear dual can be generated by words of weight at most $w$.  We will use the analogous $[[n, k, d, w]]$ for quantum stabilizer codes.  Here $w$ stands for the largest weight of a generator for the stabilizer group.  If we have a code $C$ on $n$ bits, and some subset $K \subseteq [n]$ we will use the notation $C_K$ to denote the code $C$ restricted to the bits $K$.  

Suppose $w=(w_1, w_2, \hdots{} w_n)$ is some binary vector in $\mathbb{F}_2^n$, and suppose $\{S_1, S_2, \hdots{} S_n\}$ is some set of Pauli group elements over $n$ qubits.  We will use the notation:
\be 
\prod_{i \in w} S_i : = \prod_{i \, | \, w_i=1} S_i
\ee

For matrices $A$ and $B$, we will use the standard notation $[A, B]:= A B-B A$.  The capital letters $X$, $Y$ and $Z$ will be reserved for the Pauli matrices:
\begin{equation}
X=
\begin{bmatrix}
    0 & 1 \\ 
   1 & 0
\end{bmatrix}\,\,
Y=
\begin{bmatrix}
    
     0 & -i \\ 
    i & 0
\end{bmatrix}\,\,
Z=
\begin{bmatrix}
     1 & 0 \\ 
    0 & -1
\end{bmatrix}
\end{equation}

In an $n$ qubit system, we will use $X_i$ to denote a Pauli $X$ operator acting on qubit $i$, and similarly for $Y$ and $Z$.  For a binary vector $v \in \mathbb{F}_2^n$, we define:
\be 
X_v :=\prod_{i:v(i)=1} X_i
\ee
and similarly for $Y$ and $Z$

All logarithm functions will be natural logarithm by default.  

For convenience, we will include a definition of a Bernoulli random variable to refer to later:  
\be 
{ \rm Bern}(p)=\begin{cases}
$1$ \text{ with probability } p\\
$0$ \text{ otherwise }
\end{cases}
\ee
and we will denote the binary entropy function as $h(x)$:
\be 
h(x) := -x \log_2 (x)-(1-x)\log_2(1-x)
\ee

\section{Definitions}

For us, perhaps the most important definition is that of a graph state.  We will only state the definition, for examples we invite the reader to examine any of several comprehensive reviews \cite{H04, H06}.  It may not be clear that the following definitions are equivalent a priori, proofs of equivalence can be found in the stated references.  
\begin{definition}\label{def:g_state}[Graph state]
Given some graph $G=(V, E)$ with no self loops, associate the vertices of the graph $G$ to the numbers $[1, 2, \hdots{} n]$, we define the graph state $\ket{G}$ according to three equivalent definitions:
\begin{enumerate}
\item  Let $A_{top}$ be as defined in the notations section.  We define:
\be 
\ket{G}=\frac{1}{2^{n/2}}\sum_{x \in \mathbb{F}_2^n}(-1)^{x^T A_{top}x} \ket{x}
\ee
\item Define the following Pauli group elements:
\be 
\forall i \in V \,\,\,S_i:=X_i \prod_{j \in N(i)} Z_j
\ee
The graph state $\ket{G}$ is defined as the unique state stabilized by all $S_i$.
\item \label{def:g_state:item3}  Let $CP_{ij}$ be the standard controlled phase operation between qubits $i$ and $j$.  Let $\ket{+}=\frac{1}{\sqrt{2}}\left(\ket{0}+\ket{1}\right)$.  $\ket{G}$ can be defined as:
\be 
\ket{G}=\prod_{(i, j) \in E} CP_{ij} \ket{+} \otimes \hdots{} \otimes \ket{+}
\ee

\end{enumerate}
\end{definition}
Graph states satisfy an important orthogonality property:
\begin{lemma}\label{lem:labeled_gstate}\cite{He14}
Let $x \in \mathbb{F}_2^n$ be some nonzero binary string.  Then,
\be 
\bra{G} Z_x \ket{G}=0
\ee  
\end{lemma}
\begin{proof}
Since $Z$ operators commute with controlled phase operators, according to \cref{def:g_state} we calculate:
\begin{equation} 
\bra{G} Z_x \ket{G}=
\end{equation}
\begin{equation}\label{eq:75}
\bra{+} \hdots{}\bra{+} \prod_{ij}CP_{ij} Z_x \prod_{ij} CP_{ij}\ket{+} \hdots{} \ket{+}
\end{equation}
\begin{equation}\label{eq:76}
 =\bra{+} \hdots{} \bra{+} Z_x \ket{+} \hdots{} \ket{+}=0
\end{equation}
Where we used $CP_{ij}CP_{ij}=\mathbb{I}$ to get from \cref{eq:75} to \cref{eq:76}.
\end{proof}

With this definition in hand we can define a quantum graph code.  
\begin{definition}[\cite{S01}]\label{def:graph_code}
Given a graph $G$ and a $[n, k, d]$ classical code $C$ over $\mathbb{F}_2$, we define a graph code $(G, C)$ as the linear span of quantum states of the form:
\be 
Z_c \ket{G}
\ee
where $c$ is any binary code word in $C$.
\end{definition}

Now we will present a few simple facts regarding this definition.  
\begin{lemma}
Let $(G, C)$ be as defined in \cref{def:graph_code} and suppose it has parameters $[[n, k_Q, d_Q, w_Q]]$.  Suppose the code $C$ has parameters $[n, k, d, w]$.  Denote the maximum vertex degree in $G$ as $K_{max}$ and the minimum vertex degree as $K_{min}$.

\begin{enumerate}
\item\label{def:graph_code:item1} $k=k_Q$
\item\label{def:graph_code:item2}\cite{KDP11} Let $\{h_1, \hdots{} h_{n-k}\}$ be some minimal weight generating set for the code $C^\perp$, and define:
\be 
g_j:= \prod_{i \in h_j} S_i
\ee
The code $(G, C)$ is a stabilizer code code with stabilizer generators $\{g_j\}$ for all $j$.
\item\label{def:graph_code:item3} \cite{KDP11} $w_Q \leq w K_{max}$
\item\label{def:graph_code:item4} \cite{KDP11} $d_Q \leq K_{min}$
\end{enumerate}
\end{lemma}
\begin{proof}
We can see \cref{def:graph_code:item1} immediately from \cref{lem:labeled_gstate} and \cref{def:graph_code}.

We can prove \ref{def:graph_code:item2} as follows.  We claim that the set $\{g_j\}$ form a complete, minimal set of generators for the stabilizer group for the code defined as the span of all $Z_c \ket{G}$.  Independence of these operators follows from independence of the binary vectors $h_i$.  They also all clearly commute since $[S_i, S_j]=0$ for all $i$ and $j$.  It remains to show that these operators stabilize the code.  The code $(G, C)$ is the span of states of the form $Z_c \ket{G}$ where $c \in C$.  Suppose $g_j$ has the form:
\be 
\prod_{i \in h_j} S_i=(-1)^{\phi} Z_{w} \left(\prod_{j \in h_k} X_j\right) 
\ee
for some binary vector $w$.  We have just rewritten the operator so that the $X$ operators are on the right and the $Z$ operators on the left.  Potentially we have introduced a phase $\phi\in \{0, 1\}$.
\be 
\prod_{i \in h_j} S_i Z_c \ket{G}= (-1)^{\phi} Z_{w} \left(\prod_{i \in h_j} X_i\right)  Z_c \ket{G}
\ee
Since $h_j \in C^\perp$, 
\be 
\left[ \prod_{i \in h_j} X_i, Z_c \right]=0
\ee
so
\begin{align}
(-1)^\phi Z_{w} \left(\prod_{i \in h_j} X_i\right)  Z_c \ket{G}=\\
\nonumber Z_c (-1)^\phi Z_{w} \left(\prod_{i \in h_j} X_i\right) \ket{G}=Z_c \ket{G}
\end{align}

\Cref{def:g_state:item3} follows from \cref{def:graph_code:item2}.  The stabilizers are given, and the weight of each stabilizer is upper bounded (by definition) by the maximum vertex degree times the maximum weight of $h_j$.

We sketch the proof of \cref{def:graph_code:item4}.  Focusing on the bit in the graph with smallest degree, an adversary can ``disconnect'' this vertex from the rest of the graph by enacting a unitary on this bit and its neighbors.  Then, the adversary can induce undetectable phase on the code by acting a Pauli on the disconnected bit.  Hence the adversarial distance is upper bounded by the minimal vertex degree.

\end{proof}
The classical erasure channel ``erases'' any particular message bit with probability $p$.   This corresponds to replacing that symbol with a special erasure symbol $e$.  On the other side of the channel, the user knows exactly which bits were erased and knows the other bits were unaltered by the channel.  The quantum erasure channel is the natural quantum analog of this channel.  Each transmitted bit is replaced with a special erasure state $\ket{e}$.  The state $\ket{e}$ is orthogonal to both $\ket{0}$ and $\ket{1}$ so that the user can unambiguously determine which bits were erased and which bits were unaffected by the channel.  

\begin{definition}
The erasure channel with probability $p$ will be denoted as $\mathcal{E}_p$.  It acts on density matrices of qubits as:
\be 
\mathcal{E}_p(\rho)=(1-p)\rho +p \ket{e}\bra{e}
\ee
where the state $\ket{e}$ is outside the qubit space ($\langle e|0\rangle=0=\langle e | 1\rangle$).
\end{definition}
It has been well known for some time that the capacity of this channel for communicating quantum information is $1-2p$:
\begin{lemma}\cite{CDS97}
For any $R<1-2p$, and for any $\delta > 0$, there exists a family of subspaces $\{V_n\}_{n=1}^\infty$ satisfying $\{V_n\}\subset [\mathbb{C}^2]^{\otimes n}$ and $\log_2(\dim(V_n))=\lfloor R n \rfloor$ as well as a family of decoding operations $\{\mathcal{R}_n\}_{n=1}^\infty$ such that for all $n$ sufficiently large:
\be 
\forall \ket{\psi} \in V_n, \left|\mathcal{R}_n\left( \mathcal{E}_p^{\otimes n}\left( \ket{\psi}\bra{\psi} \right)\right)-\ket{\psi}\bra{\psi}  \right|<\delta
\ee
in trace distance.

In addition, any family of codes with rate exceeding $1-2p$ is not correctable after being sent through the erasure channel\cite{MW14}\cite{W14}.  
\end{lemma}

\section{Preliminaries}

With these definitions in hand, we can present the lemmas we will need to prove that our construction approaches the erasure channel.  First, we will need several properties of graph states:
\begin{lemma}\label{lem: g_state}
Let $G=(V, E)$ be a graph on $n$ vertices.  Further, let $K$ be some subset of the bits with complement $V \setminus K$.  Let $A_{cut}$ be the cut matrix across the cut $(K, V \setminus K)$, let $A_K$ be the upper half of the adjacency matrix restricted to the bits $K$, and let $G'$ be the subgraph induced by the vertices $V \setminus K$.
\begin{enumerate}
\item  \cite{He14}Suppose a user measures the bits $K$ in the computational basis.  Every bit string $y$ is equally likely and after the user measures and gets bit string $y$, the resulting quantum state is:
\begin{equation} \label{eq:measure_subset}
\bra{y}_K \ket{G}=\frac{1}{2^{|K|/2}}(-1)^{y^T A_K y}Z_{A_{cut} y} \ket{G'}
\end{equation}
\label{item:2}
\item  \cite{H04}  Denote the entanglement entropy of the graph state across the cut $(K, V\setminus K)$ as $E^{(K, V\setminus K}(\ket{G})$.  Then,
\be 
E^{(K, V \setminus K)}(\ket{G})={\rm rank}_{\mathbb{F}_2}(A_{cut})
\ee
\label{item:3}
\end{enumerate}

\end{lemma}

In addition we will also need the following technical result on the rank of randomly chosen matrices:
\begin{theorem}\label{thm:rand_mat}[Rank of Sparse Matrices, \cite{K98}]
Let $A$ be a binary $\alpha n \times n$ matrix with $\alpha < 1$a constant.   Suppose each entry of $A$ is sampled independently according to $Bern\left(\frac{w \log(n)}{n}\right)$.  The expected number of linear combinations of rows that add to $\mathbf{0}$, or the expected number of critical sets is
\be 
1+O\left(\frac{1}{n^{w-1}} \right)
\ee
Excluding the ``all zeros'' linear combination or the empty set, the number of critical sets is:
\be 
O\left(\frac{1}{n^{w-1}} \right)
\ee
\end{theorem}
\begin{proof}
Slight modifications of the proof of Lemma 3.3.2 in \cite{K98} yield the result.  Note for reference that 
\be 
1-\frac{2 w \log(n)}{n}\leq 1-\frac{2 \log(n)}{n}
\ee
so the upper bounds present in the proof follow immediately.
\end{proof}
Hence, by Markov's inequality, if a matrix $A$ is sampled as described above, then the probability that $A$ is not full rank is upper bounded by $O\left( \frac{1}{n^{w-1}} \right)$.  Combined with the Rank Nullity theorem from linear algebra, this theorem says that if we randomly sample a parity check matrix from this ensemble, then the rows will be linearly independent with high probability.  Hence, with high probability a random $m\times n$ parity check matrix $H$ will yield a code with dimension $n-m$.  We will need a similar result on the distance of a code sampled from this ensemble.  We show that such a code has linear distance with high probability, and that the code retains this linear distance after the erasure channel:
\begin{restatable}[Linear distance]{lemma}{disterased}\label{lem:dist_erased}
Let $H$ be a $\alpha n \times n$ binary matrix with $\alpha <1$, where each entry is chosen uniformly at random according to i.i.d. $\text{Bern}(q)$ where $q=\frac{w \log(n)}{n}$. Let $C$ be the code with $H$ as its parity check matrix.  Further, let $d_{C}$ be the distance of the code $C$.

Let $K$ be the first $\beta n$ bits (corresponding to the first $\beta n$ columns of $H$) and assume that $\alpha > \beta$.  Denote the code $C$ restricted to the bits $V\setminus K$ as $C_{V\setminus K}$.  If $\epsilon'$ satisfies:
\be 
(1-\beta)h\left(\frac{\epsilon'}{1-\beta}\right) < \alpha-\beta \,\,\,\,\,\,\, \text{and} \,\,\,\,\,\,\, h(\epsilon') < \alpha
\ee
then,
\be 
d_{C_{V\setminus K}}> \epsilon' n
\ee
with probability at least
\be 
1-O\left(\frac{1}{n^{w\alpha-2}}\right)
\ee
\end{restatable}
\begin{proof}
See Appendix.
\end{proof}

We will also require the following simple result on the sum of binary random variables:  
\begin{lemma}[\cite{G63} Lemma 4.1]
Let $\{B_i\}$ be a set of $k$ independent $Bern(p)$ random variables.  Then we have that
\begin{align}
\mathbb{P} \left[\sum_{i=1}^k B_i =1 \,\, \mod \,\, 2 \right]=\frac{1-(1-2p)^k}{2}
\end{align}
\end{lemma}

%
This establishes the following corollary:

\begin{corollary}\label{cor:rand_sum_vectors}
Let $\{b_i\}$ be a set of $k$ vectors of random variables such that all entries of each $b_i$ are independent $Bern(p)$ random variables.  
For any fixed word $c \in \mathbb{F}_2^n$ such that $|c|=g$,
\begin{align} 
\mathbb{P} \left[\sum_i b_i = c \right]=\\
\nonumber \frac{1}{2^n}\left[ 1+(1-2p)^{k}\right]^{n-g} \left[ 1-(1-2p)^{k}\right]^{g}
\end{align}
\end{corollary}

We will need the ``Principle of Implicit Measurement'' to analyze the correcting properties of our graph codes.  It is standard in quantum information theory, and it states that ``lost'' quantum systems can be assumed to be measured in some basis of our choosing.  The user can treat the system as though it was measured, but does not know the outcome.  
\begin{theorem}\label{thm:implicit_measurement}[Principle of Implicit measurement]
Let $\ket{\psi}=\sum_{i, j} c_{i, j} \ket{i}_K \ket{j}_{V\setminus K}$ be a pure quantum state on $n=|V|$ qubits.  If the qubits $K$ are lost, then the reduced density matrix on qubits $V \setminus K$ is given by the partial trace:
\begin{align}
\rho_{V \setminus K} =Tr_{K} (\ket{\psi}\bra{\psi})=\\
\nonumber \sum_{i} \left( \bra{i}_K \otimes \mathbb{I}_{V \setminus K}\right) \ket{\psi} \bra{\psi} \left( \ket{i}_K \otimes \mathbb{I}_{V \setminus K}\right)
\end{align}
\end{theorem}

%
%

Lastly, we will state the Chernoff and Markov bounds for later use:
\begin{lemma}\label{lem:chernoff}[Chernoff and Markov Bounds]
Let $X$ be a random variable.  Then, we have:
\be 
\mathbb{P} (X \geq a) \leq \frac{\mathbb{E}(e^{t\cdot{} X})}{e^{t\cdot{} a}}\,\,\,\,\,\,\,\,\,\,\,\,\, \mathbb{P}(X \geq a) \leq \frac{\mathbb{E}(X)}{a}
\ee
\end{lemma}
We will make use of a simple corollary to the Chernoff bound:
\begin{corollary}\label{cor:chernoff}
Let $X=\sum_{i=1}^n X_i$ be a random variable with average $\mathbb{E}(X)=\mu$ such that the variables $X_i$ are independent.  Then, it holds that:
\be 
\mathbb{P} \left(\left|X-\mu \right| \geq \gamma \mu \right) \leq 2 e^{-\frac{\gamma^2\mu}{3}}
\ee

\end{corollary}

\section{polylog-LDPC Codes}
\subsection{Coset Measurement}

For our recovery scheme, we will require the notion of a ``coset'' measurement.  Such a measurement allows us to distinguish between different labeled graph states that are in different cosets of $C$.  Consider $C$ as a subgroup of $\mathbb{F}_2^n$.  For each vector $e \in \mathbb{F}_2^n$ we have the coset $C+e$.  Let $\{C+e\}$ be the set of all cosets of the code $C$.  Define the following subspaces:
\be 
V_e=\text{span}_{c \in C} Z_{c+e}\ket{G}
\ee
Further, define $P_e$ to be the projection onto this subspace.  We define the observable:
\be 
\mathcal{O}=\sum_{\{C+e\}} \lambda_e P_e
\ee
where the sum goes over distinct cosets of $C$ and for $e\neq e'$ we have $\lambda_e\neq \lambda_{e'}$.  It is well known that distinct cosets are disjoint, so $P_eP_{e'}=0$ for $e \neq e'$.  Hence, such an observable is well defined and serves to distinguish cosets.

\subsection{Recovery Operation}\label{sec:recovery}
Suppose we send the state $\ket{\psi} \in (G, C)$ through an erasure channel and $pn$ bits are erased.  Denote the dimension of the code $C$ as $Rn$, and the set of erased bits $K$.  Construct the following $\left[(R+p)n \right] \times (1-p)n$ matrix $F$.  Let the first $Rn$ rows of $F$ be the generators of the code $C$ restricted to the non erased bits, and let the remaining rows be the transpose of the cut matrix between the erased bits and the non-erased bits:
\begin{equation}
F= \,\,\,{\tiny {(R+p)n}} \Bigg \updownarrow \overset{ \xleftrightarrow{ \,\,\,\,(1-p) n \,\,\,\, } }{ \begin{bmatrix}
\,& C_{V\setminus K}& \, \\
\, & A_{cut}^T & \,
\end{bmatrix}}
\end{equation}
Further, let $G'$ be the subgraph of $G$ induced by the vertex set $V\setminus K$.
\begin{lemma}
If the matrix $F$ is full rank over $\mathbb{F}_2$, then there is a quantum operation $\mathcal{R}$ which recovers from the erasure.
\end{lemma}
\begin{proof}
Denote the quantum state before the channel as:
\be 
\ket{\psi}=\sum_{c \in C} b_c Z_c \ket{G}
\ee
Suppose without loss of generality that the erased bits, $K$, are the first $pn$ bits (or rearrange the bits so that this is the case).  For each codeword $c\in C$, decompose $c$ as the concatenation of its value on the erased bits with its value on the non erased bits: $c=(c_K, c_{V\setminus K})$.  By hypothesis, and by \cref{lem: g_state}, after the erasure channel, we are left with the following density matrix on the subsystems $V\setminus K$:
\begin{equation}\label{eq:41} 
\rho=\sum_{j\in \mathbb{F}_2^{|K|}}\rho_j \ket{\phi_j} \bra{\phi_j}
\end{equation}
where:
\be 
\sum \rho_j=1 \,\,\,\,\,\,\,\text{and}\,\,\,\,\,\,\,\ket{\phi_j}=\sum_{c \in C} b_c (-1)^{j \cdot{} c_K} Z_{c_{V\setminus K}+A_{cut}j} \ket{G'}_{V \setminus K}
\ee

If $F$ is full rank, then the coset measurement of $C_{V\setminus K}$ can be used to determine the string $j$.  Indeed, if $F$ is full rank then each element of the range of $A_{cut}$ belongs to a different coset.  So, measuring which coset the graph state falls into will determine the string $j$.  Note further that such a measurement will not disturb the encoded information since for fixed $j$ the states $\{Z_{c_{V\setminus K}+A_{cut} j}\}$ all lie in a particular coset of $C_{V\setminus K}$.  

Let us describe this more formally.  Suppose that $A_{cut}j$ falls into a particular coset of $C_{V\setminus K}$:
\begin{equation}\label{eq:42}
A_{cut}j =c_{V\setminus K}''+e
\end{equation}
for some fixed $c_{V\setminus K}'' \in C_{V\setminus K}$.  If $P_{e'}$ is the projector onto a different coset, then $P_{e'}\ket{\phi_j}=0$:
\be 
P_{e'}\ket{\phi_j}=\sum_{c, c' \in C} \sigma_{(c, c', e', j)}\bra{G'} Z_{A_{cut}j+c_{V\setminus K}+c'_{V\setminus K}+e'}\ket{G'}
\ee
for some operators $\sigma_{(c, c', e', j)}$.  Since $e'$ corresponds to a distinct coset, the string $ A_{cut}j+c_{V\setminus K}+c'_{V\setminus K}+e'$ is always in some nonzero coset, so by \cref{lem:labeled_gstate}, $P_{e'}\ket{\phi_j}=0$.

The projector $P_e$ has the following effect on the state $\ket{\phi_j}$:
\begin{align}
P_e\ket{\phi_j}=\left[ \sum_{c_{V\setminus K}'}Z_{c_{V\setminus K}'+e}\ket{G'}\bra{G'} Z_{c_{V\setminus K}'+e}\right]\ket{\phi_j}\\
\nonumber = \sum_{\substack{c_{V\setminus K}', c_{V\setminus K} : \\ c'_{V\setminus K} +e =c_{V\setminus K}+A_{cut} j}} b_c (-1)^{j\cdot{} c_K} Z_{c_{V\setminus K}'+e}\ket{G'} 
\end{align}
The relation $c'_{V\setminus K}+e =c_{V\setminus K}+A_{cut}j$ fixes $c_{V\setminus K}'$ given $c_{V\setminus K}$.  It is easy to see that we get back exactly the state $\ket{\phi_j}$.  

To complete our description of the recovery operation $\mathcal{R}$ we need to give the unitary that can recover the original quantum state $\ket{\psi}$ given the state $\ket{\phi_j}$ and given the string $j$.   By appending extra copies of the state $\ket{+}$, we can achieve the state:
\be 
\sum_{c \in C} b_c (-1)^{j \cdot{} c_K} Z_{c_{V\setminus K}+A_{cut} j}\ket{+\hdots{} +}_K \ket{G'}_{V\setminus K}
\ee 
Now we an apply controlled phase operations (using \cref{def:g_state}) to transform the state to:
\be 
\sum_{c\in C} b_c (-1)^{j\cdot{} c_K} Z_{c_{V\setminus K}+A_{cut} j} \ket{G}_V
\ee
Now we can apply the unitary $Z_{A_{cut}j}$ to transform the state to:
\be 
\sum_{c\in C} b_c (-1)^{j \cdot{} c_K} Z_{c_{V\setminus K}}\ket{G}
\ee
Since 
\be 
\langle G| Z_{c_{V\setminus K}} Z_{c'_{V\setminus K}} | G \rangle=0=\langle G | Z_{c} Z_{c'} | G \rangle
\ee
the unitary that sends $Z_{c_{V\setminus K}}\ket{G} \rightarrow Z_c \ket{G}$ is well defined.  We can apply it to obtain:
\be 
\sum_{c\in C} b_c (-1)^{j \cdot{} c_K} Z_c \ket{G}
\ee 
Finally we can apply a unitary that is diagonal in the $Z_c \ket{G}$ basis to remove the phase and recover $\ket{\psi}$
\end{proof}

To leverage the previous result we have to show that the matrix $F$ will be full rank with high probability:
\begin{theorem}
For any $0.33 < p <\frac{1}{2}$, let $R < 1-2 p$.  Also fix $q=\frac{w\log(n)}{n}$ for some constant $w>1$.  Let $H$ be a randomly sampled $(1-R)n \times n$ matrix where all $H_{ij}$ are distributed according to i.i.d $\text{Bern}(q)$ random variables.  Denote the code with parity check matrix $H$ as $C$.  Let $G$ be a randomly sampled graph where we begin with an empty graph and add each edge independently with probability $q$.  

In analogy to the erasure channel, suppose each column is `erased' with probability $p$.  I.e. we start with the empty set $K\subseteq [n]$ and add each bit independently to $K$ with probability $p$.  Let $F$ be the matrix defined in \cref{sec:recovery} given the subset $K$.  The probability that $F$ is not full rank satisfies:
\be 
p_e=O\left( \frac{1}{n^{w-1}}+\frac{1}{n^{2(pw-1)}}\right)
\ee
\end{theorem}
\begin{proof}
Let $K$ denote the set of erased bits.  By \cref{cor:chernoff}, we can assume that $|K|=p'n \in [(p-\delta')n, (p+\delta')n]$ with probability exponentially close to $1$ for any $\delta' >0$.  Since the erased bits are independent of the code, we can assume without loss of generality that $K$ consists of the first $p' n$ bits and randomly sample our code.  Let us define the constant $\delta$ such that $R+\delta=1-2 p $.  The matrix $F$ is as defined in the previous lemma:
\be 
F= \,\,\,{\tiny {(R+p')n}} \Bigg \updownarrow \overset{ \xleftrightarrow{ \,\,\,\,(1-p') n \,\,\,\, } }{ \begin{bmatrix}
\,& C_{V\setminus K}& \, \\
\, & A_{cut}^T & \,
\end{bmatrix}}
\ee
where $C_{V\setminus K}$ are the generators of the code $C$ restricted to the non-erased bits, and $A_{cut}$ is the $|V\setminus K|\times|K|$ cut matrix across the cut $(K, V\setminus K)$.  The matrix $F$ fails to to full rank if and only if one the the following events occurs:
\begin{enumerate}
\item $A_1$- The rows of $C_{V\setminus K}$ are linearly dependent
\item $A_2$- The rows of $A_{cut}^T$ are linearly dependent
\item $A_3$- A linear combination of the rows of $A_{cut}^T$ produce some word in $C_{V\setminus K}$
\end{enumerate} 
We will produce upper bounds on the probabilities of each of these events, and use the union bound to find an upper bound on the probability that $F$ is not full rank.

For $A_1$, we will argue using the randomly generated parity check matrix of the classical code.  Recall that the parity check matrix $H$ is a $[(1-R)n] \times n$ matrix such that each entry is $1$ with probability $q$, and $0$ otherwise.  Denote the codewords $c\in C$ as $c=(a, b)$ where $a$ is supported on the erased bits and $b$ is supported on the non-erased bits.  Further, let $\{(a_i, b_i)\}$ be a minimal set of generators for the code.  Now observe the following equivalence: The set of vectors $\{b_i\}$ is linearly dependent if and only if there is some nonzero $c=(a, b) \in C$ with b=0.
Observe that, if $H$ is full rank on the first $p'n$ bits, then there is no non-zero $c \in C$ such that $c=(a, b)$ with $b=0$.  Hence, the probability $\mathbb{P}(A_1)$ is less than or equal to the probability that the first $p'n$ columns of $H$ are linearly dependent.

The first $p'n$ columns of $H$ correspond to a random $[(1-R)n]\times p' n=(2p +\delta)n \times p' n$ matrix where each entry is $1$ with probability $q$.  Note that, conditioned on the erased bits, we can treat this submatrix as independently sampled as in \cref{thm:rand_mat}.  If $\delta'$ is small compared to $p$, then this submatrix has more rows than columns by some constant fraction of $n$.  Hence, we derive the following upper bound using \cref{thm:rand_mat}
\be 
\mathbb{P}(A_1) =O\left( \frac{1}{n^{w-1}} \right)
\ee

We can bound $\mathbb{P}(A_2)$ directly using \cref{thm:rand_mat}.  $A_{cut}^T$ is a randomly sampled $p'n \times (1-p')n$ binary matrix where each entry is $1$ independently with probability $\frac{w\log(n)}{n}$.  Since $p<\frac{1}{2}$, there are more rows of $A_{cut}^T$ than there are columns by some constant fraction of $n$ if we take $\delta'$ small enough.  Hence:
\be 
\mathbb{P}(A_2) =O\left(\frac{1}{n^{w-1}}\right)
\ee

Bounding $\mathbb{P}(A_3)$ requires the technical result proven in the appendix.  Let $B_1$ be the event `$d_{C_{V\setminus K}} \leq \epsilon' n$' where $\epsilon' =H^{-1}(p)$.  We write:
\be 
\mathbb{P}(A_3)= \mathbb{P}(A_3\cap B_1) +\mathbb{P}(A_3 \cap \neg B_1)
\ee
Where we have used the negation symbol $\neg$ to indicate the complement.  It is easy to check that we have met the conditions required for \cref{lem:dist_erased} for small enough $\delta'$:
\be 
(1-p')H\left(\frac{H^{-1} (p)}{1-p'}\right) < 2p+\delta-p-\delta'
\ee
and
\be 
H(\epsilon') =p < 2 p+\delta
\ee
So, we can upper bound:
\be 
\mathbb{P}(A_3\cap B_1) \leq \mathbb{P}(B_1) =O\left( \frac{1}{n^{(2p+\delta)w-2}} \right)=O\left( \frac{1}{n^{2(p w-1)}}\right)
\ee
Now we need to find an upper bound on $\mathbb{P}(A_3 \cap \neg B_1)$.  Let $c_{V\setminus K}$ be some word in $C_{V\setminus K}$ of weight $g$, and let $v $ be any word in $\mathbb{F}_2^{p' n}$ of weight $k$.  By \cref{cor:rand_sum_vectors}, we write:
\be 
\frac{b(k)}{2^{(1-p')n}} := \mathbb{P}\left(A_{cut}v =c_{V\setminus K} \right)=\frac{1}{2^{(1-p')n}}\left(1+\left(1-\frac{2 w \log(n)}{n} \right)^k \right)^{(1-p')n-g}\left(1-\left(1-\frac{2 w \log(n)}{n} \right)^k \right)^g
\ee

It is sufficient for an upper bound to analyze the word $c_{V\setminus K}$ of smallest weight, since clearly this expression is decreasing with $g$ increasing.  Since $d_{C_{V\setminus K}} > \epsilon' n$ we can assume $g=\epsilon' n$.  Let us define the intervals $I_1=\left[1, \frac{z n}{\log(n)}\right]$ and $I_2=\left[\frac{z n }{\log(n)}, p' n\right]$ for some large constant $z$ to be determined.  For fixed $g$, we will first provide an upper bound for this function inside the interval $I_1$.  Let us analyze the derivative of $b(k)$ with respect to $k$:

\begin{align}
b'(k)=\left[1+\left(1-\frac{2 w \log(n)}{n} \right)^k \right]^{(1-p')n-g-1} \left[1+\left(1-\frac{2 w \log(n)}{n} \right)^k \right]^{g-1}\log\left(1-\frac{2 w \log(n)}{n} \right) \times\\
\nonumber \left\{ - \left(1+\left(1-\frac{2 w \log(n)}{n} \right)^k \right)g+ ((1-p')n-g)\left(1-\left(1-\frac{2 w \log(n)}{n} \right)^k \right)\right\}
\end{align}
We can set this expression equal to zero and solve.  We obtain:
\be 
k_{max}^1=\frac{\log\left(\frac{(1-p')n-2g}{(1-p')n} \right)}{\log\left(1-\frac{2 w \log(n)}{n} \right)}=\frac{\log\left(1-\frac{2 \epsilon '}{1-p'}\right)}{\log\left(1-\frac{2 w \log(n)}{n} \right)}
\ee
The function $b'(k)$ is peaked around $k_{max}^1$, or $b'(k) \geq 0 $ for $k \leq k_{max}^1$ and $b'(k)\leq 0$ for $k \geq k_{max}^1$.  Expanding the expression for $k_{max}^1$, we can see that $k_{max}^1 \in I_1$, so $b(k_{max}^1)$ provides an upper bound for $b(k)$ in this interval.  We calculate:
\begin{align}
\frac{b(k_{max}^1)}{2^{(1-p')n}} =\frac{1}{2^{(1-p')n}} \left[2-2\frac{\epsilon'}{1-p'}\right]^{(1-p')n-\epsilon' n} \left[\frac{2 \epsilon '}{1-p'}\right]^{\epsilon ' n}\\
\nonumber = 2^{\wedge} \left[-\epsilon' n +\log_2\left(1-\frac{\epsilon'}{1-p'}\right) ((1-p')n-\epsilon' n)+\log_2\left(\frac{2\epsilon '}{1-p'}\right)\epsilon ' n\right]
\end{align}

In the second interval $I_2$, the best upper bound we can obtain is:
\be 
\frac{b\left(\frac{z n}{\log(n)}\right)}{2^{(1-p')n}} \leq \frac{1}{2^{(1-p')n}}\left[1+\left(1-\frac{2 w \log(n)}{n} \right)^{\frac{z n}{\log(n)}}\right]^{(1-p')n}
\ee
For large enough $n$, 
\be 
\leq \frac{1}{2^{(1-p')n}}\left[1+e^{-2 w z}\right]^{(1-p')n}=:\frac{b(k_{max}^2)}{2^{(1-p')n}}
\ee

We are interested in finding an upper bound for the probability that any $v\in \mathbb{F}_2^{p' n}$ maps to any $c_{V\setminus K} \in C_{V\setminus K}$ under $A_{cut}$.  For this we can employ the union bound:
\be 
\mathbb{P}(\exists v \in \mathbb{F}_2^{p' n}, \exists c_{V\setminus K} \in C_{V\setminus K}: A_{cut} v=c_{V\setminus K} \cap \neg B_1) \leq \sum_{\substack{v\in \mathbb{F}_2^{p' n}\\ c_{V\setminus K} \in C_{V\setminus K}}} \mathbb{P}[A_{cut} v=c_{V\setminus K}\cap \neg B_1]
\ee
Under our assumptions, $c_{V\setminus K} > \epsilon' n$, we can further upper bound this expression by:
\begin{align} 
\leq 2^{R n}\left(\sum_{|v| \in I_1} \frac{b(k_{max}^1)}{2^{(1-p')n}}+\sum_{|v|\in I_2}\frac{b(k_{max}^2)}{2^{(1-p')n}}\right)
\end{align}
where we used the fact that the code $C_{V\setminus K}$ has at most $2^{Rn}$ many words.  We then note that there are at most $2^{o(n)}$ many terms in the first sum.  The upper bound we obtain is:
\be 
\frac{2^{Rn}2^{o(n)} b(k_{max}^1)}{2^{(1-p')n}}+\frac{2^{Rn}2^{p'n}}{2^{(1-p')n}}\left[1+e^{-2 w z}\right]^{(1-p')n}
\ee
For large enough $z$ and small enough $\delta'$ the second term is exponentially small with $n$.  The first term is exponentially small if:
\begin{equation}\label{eq:67} 
R-\epsilon'+\log_2\left(1-\frac{\epsilon'}{1-p'}\right)\left((1-p')-\epsilon' \right)+\log_2\left(\frac{2\epsilon'}{1-p'}\right)\epsilon' < 0
\end{equation}
If 
\be 
g(p) := 1- 2p -\epsilon'+\log_2\left(1-\frac{\epsilon'}{1-p}\right)\left((1-p)-\epsilon' \right)+\log_2\left(\frac{2\epsilon'}{1-p}\right)\epsilon' < 0
\ee
Then we can make $\delta'$ small enough that \cref{eq:67} holds.  We have found computationally $g(p) < 0$ for $0.33 <p < \frac{1}{2}$ (recall that we set $\varepsilon'=h^{-1}(p)$).
\end{proof}

\begin{corollary}[Quantum Erasure Channel]
For any $0.33 < p< \frac{1}{2}$, let $R < 1-2 p$.  Let $q=\frac{w\log(n)}{n}$ for some constant $w>1$.  Let $H$ be a randomly sampled $(1-R)n \times n$ matrix where all $H_{ij}$ are distributed according to i.i.d. $\text{Bern}(q)$ random variables.  Denote the code with parity check matrix $H$ as $C$.  Let $G$ be a randomly sampled graph where we begin with an empty graph and add each edge independently with probability $q$.  

The quantum code $(G, C)$ has vanishing probability of making a decoding error when sent over the quantum erasure channel $\mathcal{E}_p$.  In particular, the probability of error satisfies:
\be 
p_e =O\left( \frac{1}{n^{w-1}}+\frac{1}{n^{2(pw-1)}}\right)
\ee
\end{corollary}

\subsection{Size Bounds on the code $(G, C)$}
Now to provide a probabilistic estimate on the weight of the randomly chosen stabilizer code.  For $i \in [1, \hdots{} (1-R)n]$, define the random variable $X_i$ to be the weight of the $i$th row of the parity check matrix $H$.  By \cref{lem:chernoff}, for each $i$
\begin{align}
\mathbb{P}(X_i \geq \log^{1+\zeta}(n)) \leq \frac{\mathbb{E}(e^{t \cdot{} X_i })}{e^{t \cdot{} a}}=\\
\nonumber \frac{ e^{n\left(\frac{w \log(n)}{n} \right)(e^t-1)} }{e^{t \log^{1+\zeta}(n)}}=\frac{n^{w(e^t-1) } }{n^{t \log^\zeta(n)}}
\end{align}
So, we can calculate via the union bound:
\begin{align}
\mathbb{P}\left( \cup_{i} \{X_i \geq \log^{1+\zeta}(n)\} \right)\leq \frac{n e^{n\left(\frac{w \log(n)}{n} \right)(e^t-1)} }{e^{t \log^{1+\zeta}(n)}}=\\
\nonumber \frac{n^{w(e^t-1)+1 } }{n^{t \log^\zeta(n)}}
\end{align}
Now for $i \in [1, 2, \hdots{} n]$ define the random variable $Y_i$ to be the number of neighbors of a vertex $i$ in the randomly generated graph $G$.  The same analysis yields:
\be 
\mathbb{P}\left( \cup_{i} \{Y_i \geq \log^{1+\zeta}(n)\} \right)\leq \frac{n^{w(e^t-1)+1 } }{n^{t \log^\zeta(n)}}
\ee
If all $X_i$ and $Y_i$ are less than $\log^{1+\zeta}(n)$, then the maximum weight of a generator of the stabilizer quantum code is upper bounded by $\log^{2+2\zeta}(n) $.  

\be
\mathbb{P} (\text{code $(G, C)$ is a $[[n, k, d, j]]$ code} \text{ with $j \geq \log^{2+2\zeta}(n)$}) \leq \frac{2 n^{w(e^t-1)+1 } }{n^{t \log^\zeta(n)}}
\ee
If we take $t=O(1)$ we obtain vanishing probability with $n$ for any $\zeta > 0$.

\section{Conclusion}

In this paper we have given random constructions of quantum stabilizer codes that both achieve (come arbitrarily close to) the capacity of the erasure channel, while at the same time are $\text{poly}\log$-LDPC with high probability.  We stress that these codes are interesting primarily due to the work of Delfosse et al. \cite{DZ13}, since such a property is impossible with codes that are LDPC.  

We speculate that the graph states presented have interesting entanglement properties: in a sense they are ``almost'' absolutely maximally entangled states \cite{He14}.  An absolutely maximally entangled state is a state that is maximally entangled across any partition of the subsystems.  The graph states we have described have the following property:  Choose any partition of the subsystems $(A, B)$ such that $A=\alpha n$ and $B=(1-\alpha) n$ where $\alpha<\frac{1}{2}$.  The quantum graph state described is maximally entangled across the partition with all but inverse polynomial probability in $n$.  This is easily seen from \cref{lem: g_state}, and \cref{thm:rand_mat}\footnote{Note that if we add an edge with probability $1/2$, then we could have instead used a much simpler version of \cref{thm:rand_mat} to prove this.  We required the more complicated result of Kolchin because we needed $log$-LDPC.}.  Hence, the quantum state is maximally entangled across almost all partitions where one set in the partition is larger than the other set by some constant fraction of the qubits.  Indeed, we first became interested in these states because of this property.  

\section{Acknowledgments}
The authors wish to thank Lior Eldar and Murphy Y. Niu for helpful comments on the draft.  

\bibliographystyle{plain}
\bibliography{References}

\appendix
\section{Appendix}
Before starting the proofs we present a few simple definitions and facts.  The Gamma function is defined as:
\be 
\Gamma(y):=\int_{0}^\infty x^{y-1} e^{-x} dx
\ee
The Digamma functions is defined as:
\be 
\psi^{(0)}(y):=\frac{\Gamma'(y)}{\Gamma(y)}
\ee
and the $m$th order Polygamma function is defined as:
\be 
\psi^{(m)}(y):=\frac{d^m}{dy^m}\psi^{(0)}(y)
\ee
The Harmonic numbers are defined as:
\be 
H_{x-1}:=\psi^{(0)}(x)+\gamma\,\,\,\,\,\,\,\,\,\,\,\,\,\,\, H_{x-1}^{(2)}:=\frac{\pi^2}{6}-\psi^{(1)}(x)
\ee
where $\gamma$ is the Euler Mascheroni constant.  At integer values, the Harmonic numbers have their usual expression:
\be 
H_k=\sum_{j=1}^k \frac{1}{j}\,\,\,\,\,\,\,\,\,\,\,\,\, H_k^{(2)}=\sum_{j=1}^{k} \frac{1}{j^2}
\ee
There are many important properties of these definitions, we list the properties we will need for the proofs:
\begin{fact}\label{lem:list}
\begin{enumerate}
\item  The Gamma function is equal to the factorial at positive integer arguments:
\be 
\forall j \in \mathbb{Z}, \,\, > 0:  \,\, \Gamma(j-1)=j!
\ee

\item  We can approximate $\psi^{(1)}(k)$ as $\Theta(1/k)$:
\be 
\forall k > 0:\,\,\,\, \frac{1}{k} \leq \psi^{(1)}(k) 
\leq \frac{2}{k}
\ee
which implies by definition:
\be 
\frac{\pi^2}{6}-\frac{2}{k} \leq H^{(2)}_{k-1}  \leq \frac{\pi^2}{6}-\frac{1}{k}
\ee
\item  For positive $k$:
\be 
\gamma +\log(k) < H_k < \gamma +\log(k+1)
\ee

\end{enumerate}
\end{fact}
Throughout the appendix we will have some $n$ and $w$ in mind.  We will denote:
\be 
a:=1-\frac{2 w \log(n)}{n}
\ee
The goal of the appendix is to show that codes from our ensemble have linear distance after the erasure channel.  We find it useful to first prove that the code itself has linear distance with high probability:
\begin{lemma}\label{lem:almost_linear}
Let $H$ be a $\alpha n \times n$ binary matrix with $\alpha < 1$, where each entry is chosen uniformly at random according to i.i.d. $\text{Bern}(q)$ where $q=\frac{w \log(n)}{n}$. Let $C$ be the code with $H$ as its parity check matrix.  Further, let $d_C$ be the distance of the code $C$.  For any constant $\epsilon > 0$ satisfying $h(\epsilon) < \alpha$
\be 
d_C > \epsilon n
\ee
with probability at least:
\be 
1-O\left(\frac{1}{n^{w \alpha-2}}\right)
\ee
\end{lemma}
\begin{proof}
We use first moment methods.  Let $X$ be a random variable equal to the number of subsets of columns of $H$ with size at most $\epsilon n$ that sum to zero.  $X$ can equivalently be defined as the number of words of $C$ with weight less than $\epsilon n$.  We calculate using \cref{cor:rand_sum_vectors}
\be 
\mathbb{E}(X)=\sum_{k=1}^{\epsilon n}\binom{n}{k}\left( \frac{1+a^k}{2}\right)^{\alpha n}
\ee

Let us define the function:
\be 
f(k):= \binom{n}{k}\left(1+a^k \right)^{\alpha n}
\ee
We can make this function continuous and differientiable by substituting Gamma functions for factorials:
\be 
\binom{n}{k}=\frac{\Gamma(n-1)}{\Gamma(k-1)\Gamma(n-k-1)}
\ee
We will show that we can find an upper bound for $f(k)$ by examining the endpoints $k=1$ and $k=\epsilon n$.  Define two intervals $I_1=\left[1, \frac{z n }{\log(n)}\right]$ and $I_2=\left[\frac{z n}{\log(n)}, \epsilon n\right]$ where $z$ is some large constant to be determined.  The property we claim follows if we can show that $f'(k) < 0$ in the interval $I_1$ and that the function $f'(k)$ has exactly one zero in the interval $I_2$.  The remainder of the proof will fall into two parts.  In part a we will demonstrate that $f'(k) < 0$ in $I_1$ and in part b we will demonstrate that $f'(k)$ has exactly one zero in $I_2$.

\begin{figure}
\center
  \includegraphics[width=300pt]{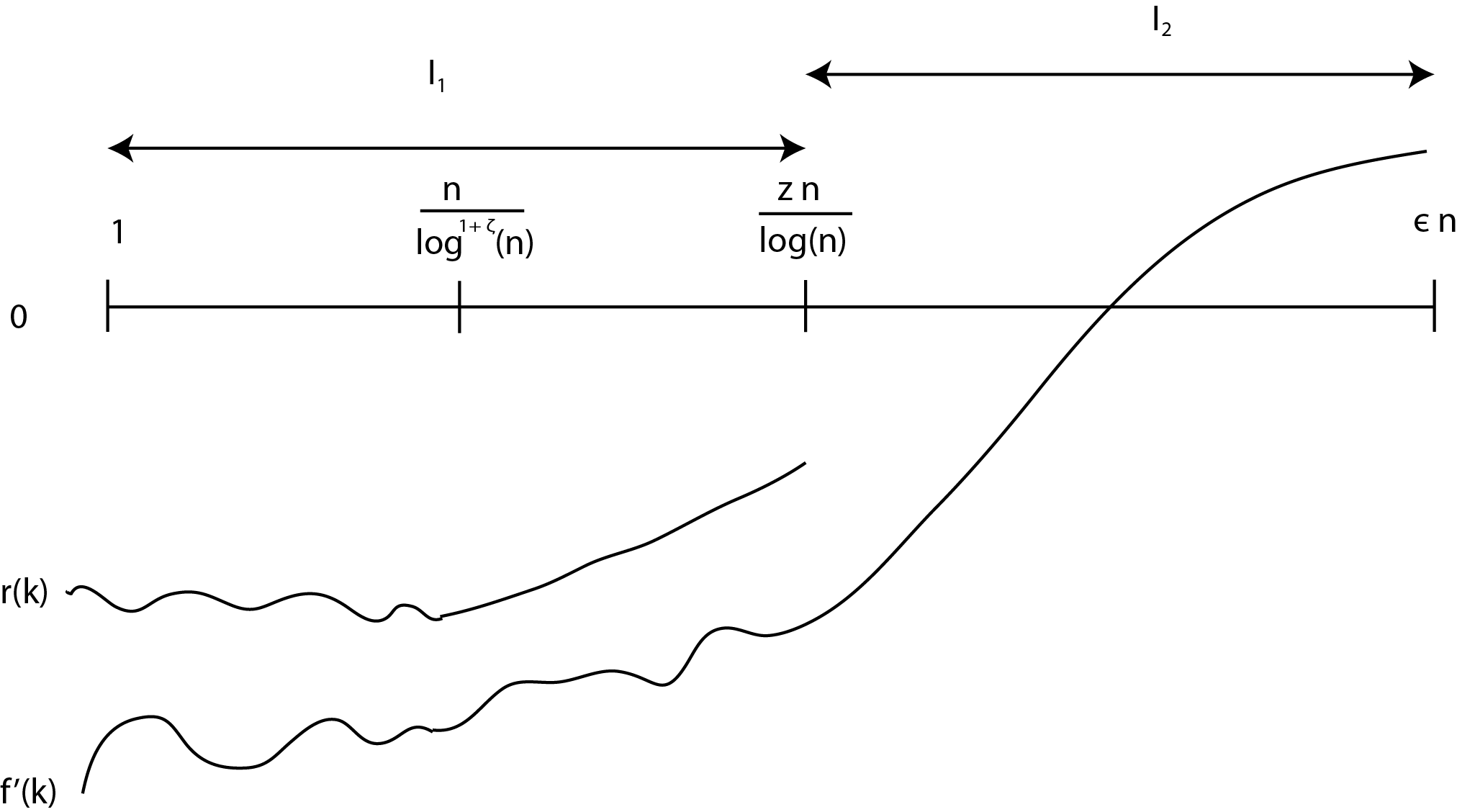}
  \caption{An illustration of our proof method. We demonstrate that $f'$ has the above form, which implies exactly one local minimum.  Therefore, $f$ must be maximized at one of the endpoints (either at $k=1$ or at $k=\epsilon n$.  }
  \label{fig:proof_method}
\end{figure}

\noindent \Large{\textbf{Part a}}

\noindent \normalsize We will further divide interval $I_1$ into two other intervals: $\left[1, \frac{n}{\log^{1+\zeta}(n)}\right]$ and $\left[\frac{n}{\log^{1+\zeta}(n)}, \frac{z n}{\log(n)}\right]$ where $\zeta$ is some small positive constant.  We will provide another function $r(k)$ which upper bounds $f'$ in both of these intervals.  We first demonstrate that $r(k) < 0$ in the interval $\left[1, \frac{n}{\log^{1+\zeta}(n)}\right]$.  Then we will show that $r'(k)$ has a positive slope in the interval $\left[\frac{n}{\log^{1+\zeta}(n)}, \frac{z n}{\log(n)}\right]$ and has a negative endpoint at $\frac{z n}{\log(n)}$, implying $f'(k) < 0$ throughout $I_1$.

We calculate:
\be 
f'(k)=v(k) \left(a^k \left(\alpha n \log
   \left(a\right)+H_{n-k}-H_k\right)+H_{n-k}-H_k\right)
\ee
where $v(k) > 0$
We are interested in the sign of $f'(k)$, so it is sufficient to study $\frac{f'(k)}{v(k)}$.  Define:
\be 
b(k):=a^k\left(\alpha n \log\left(a \right)+H_{n-k}-H_k \right)+H_{n-k}-H_k
\ee

For large enough $n$ we can expand to show:
\begin{equation}\label{eq:74}
\frac{-3 w \log(n)}{n}\leq\log\left(a \right) \leq  -\frac{2 w log(n)}{n}
\end{equation}
So,
\be 
b(k)\leq a^k\left(H_{n-k}-H_k -2 w \alpha  \log(n)\right)+H_{n-k}-H_k
\ee
By \cref{lem:list}:
\be 
H_{n-k} \leq \log(n) +\gamma
\ee
so we have $b(k) \leq r(k)$ where:
\be 
r(k):=a^k \left(-2 w \alpha +2\right) \log(n)+H_{n-k}-H_k
\ee
We need to show that the function $r(k) < 0 $ for all $k$ in the interval $\left[1, z\frac{n}{\log(n)}\right]$.  

For any $k \in \left[1, \frac{n}{\log^{1+\zeta}(n)} \right]$ the function $r(k)$ is negative for sufficiently large $n$.  Indeed for $k=\frac{n}{\log^{1+\zeta}(n)} $:
\be 
\lim_{n\rightarrow \infty} a^k=\lim_{n\rightarrow \infty} \left(1-\frac{2 w \log(n)}{n} \right)^{\frac{n}{\log^{1+\zeta}(n)}} =1
\ee
The Harmonic terms are of order $\log(n)$ at most, so the first term dominates if $w\alpha>2$.  

Now we will show that the term $r'(k) \geq 0$ for all $k$ in the interval $\left[ \frac{n}{\log^{1+\zeta}(n)}, z\frac{n}{\log(n)} \right]$.  We calculate:
\be 
r'(k)=a^k(-2 \alpha w+2)\log(n)\log \left(a \right) +H_{n-k}^{(2)}+H_k^{(2)}-\frac{\pi ^2}{3}
\ee
By \cref{eq:74}:
\be 
r'(k) \geq \frac{4 w^2\alpha \log^2(n)}{n}
   \left(a\right)^k+H_{n-k}^{(2)}+H_k^{(2)}-\frac{\pi ^2}{3}
\ee
Now we can apply \cref{lem:list}, and the fact that $\frac{1}{n-k} \leq \frac{1}{k}$ to obtain:
\be 
r'(k) \geq \frac{4 w^2\alpha \log^2(n)}{n}
   \left(1-\frac{2 w \log
   (n)}{n}\right)^k-\frac{4}{k}
\ee
By definition of our interval, we obtain:
\be 
k \geq \frac{n}{\log^{1+\zeta}(n)} \Rightarrow - \frac{1}{k}\geq - \frac{\log^{1+\zeta}(n)}{n}
\ee
For large enough $n$, we have: 
\be 
k \leq z\frac{n}{\log(n)}\Rightarrow \left(1-\frac{2 w \log(n)}{n}\right)^k \geq e^{-2 w z-1}
\ee
Hence,
\be 
r'(k) \geq\frac{4 w^2\alpha e^{-2 w z-1} \log^2(n)}{n} -\frac{4 \log^{1+\zeta}(n)}{n} > 0
\ee
Once we demonstrate that $r\left(\frac{zn}{\log(n)}\right) < 0$, we will have shown that $r(k) < 0$ for all $k\in \left[1, \frac{zn}{\log(n)}\right]$.  Indeed:
\be 
r\left(\frac{zn}{\log(n)}\right)=\left( a \right)^{\frac{zn}{\log(n)}} \left(-2w\alpha + 2\right) \log(n)+H_{n-\frac{zn}{\log(n)}}-H_{\frac{zn}{\log(n)}}
\ee
It is not hard to see that:
\be 
H_{n-\frac{zn}{\log(n)}}-H_{\frac{zn}{\log(n)}}=O(\log(\log(n)))
\ee
from \cref{lem:list}.  In addition, the term:
\be 
\left(1-\frac{2 w \log(n)}{n} \right)^{\frac{zn}{\log(n)}}=\Theta(1)
\ee
So, for large enough $n$, $r\left(\frac{zn}{\log(n)}\right) < 0$.  

We have shown that $r(k) < 0$ for all $k \in \left[1, \frac{z n}{\log(n)}\right]$.  This implies $b(k)< 0$ for all $k \in I_1$ which in turn implies that $f'(k) < 0$ for these $k$.  

\Large \textbf{Part b}

\normalsize \noindent Now we will show that inside the interval $k \in I_2$ the function $f'(k)$ has exactly one zero.  By rearranging, we can see that $f'(k)=0$ if and only if
\begin{equation}\label{eq:107} 
(H_{n-k}-H_k)\left(\frac{1}{a^k}+1\right)=-\alpha n \log\left(a\right)
\end{equation}
We define:
\be 
g(k) := (H_{n-k}-H_k )\left(\frac{1}{a^k}+1 \right)
\ee
and calculate:
\be 
g'(k)= \frac{1}{a^k} \left(\log(a) (H_{k}-H_{n-k})-(1+a^k) (\psi^{(1)}(n-k+1)+\psi ^{(1)}(k+1))\right)
\ee
So we can lower bound:
\begin{align} 
a^k g'(k) \geq \frac{2 w \log(n)}{n} \left(H_{n-k}-H_k \right) -\frac{4}{k}\\
\nonumber \geq \frac{2 w \log(n)}{n}(H_{n-k}-H_k)-\frac{4\log(n)}{z n} >0
\end{align}
which is positive for $z$ large enough since $H_{n-k}-H_k=\Omega(1)$ in this interval.  Since $g(k)$ is strictly increasing with $k$ and the RHS of \cref{eq:107} is fixed, there can be at most one place where $f'(k)=0$.  We have already shown
\be 
f'\left( \frac{z n }{\log(n)} \right) < 0
\ee
and it is not hard to see:
\be  
f'(\epsilon n) > 0
\ee
for large enough $n$ so $f'$ has exactly one zero in this interval.

Now back to the problem at hand, we want an upper bound on $\mathbb{E}(X)$.  Our analysis implies that we can use the largest endpoint as an upper bound on $f(k)$:
\begin{align} 
\mathbb{E}(X) \leq \max\left \{\epsilon n \binom{n}{1}\left( 1-\frac{2 w \log(n)}{n}\right)^{\alpha n}, \epsilon n \binom{n}{\epsilon n} \left(\frac{1+a^k}{2}\right)^{\alpha n} \right \}\\
\nonumber \leq \max \left \{ \frac{1}{n^{w\alpha -2}}, 2^{o(n)+(h(\epsilon)-\alpha) n } \right\}
\end{align}
By hypothesis the second term is exponentially small.  To complete the proof we use Markov's inequality.  Let us define the event $A_1$ as the event `$d_{C} \leq \epsilon n$'.  Then, we have:
\be 
\mathbb{P}(A_1) =\mathbb{P}(X\geq 1) \leq \mathbb{E}(X) =O\left(\frac{1}{n^{w \alpha-2}} \right)
\ee
\end{proof}

The result we are actually interested involves the code $C$ after the erasure channel.  We need to show that this code still has linear distance.  Suppose as in the paper that the set of bits $K$ has been erased, and say that the size of this set is $\beta n$ for some constant $\beta$.  We establish the following lemma:

\disterased*

\begin{proof}
We will use the previous lemma to show that we can expect the code $C$ to have linear distance, and condition on this event to show that the code $C_{V\setminus K}$ satisfies the same property with a weaker constant.  

Again let $A_1$ be the event that $d_C \leq \epsilon n$ where $\epsilon > \epsilon'$ and $H(\epsilon) < \alpha$.  Further, let $A_2$ be the event that $d_{C_{V\setminus K}} \leq \epsilon ' n $.  Just as in the paper, we will use the negation symbol $\neg$ for the complement.  So $\neg A_1$ is the event that $d_C > \epsilon n$.  We can write:
\be 
\mathbb{P}(A_2)=\mathbb{P}(A_2 \cap A_1) +\mathbb{P}(A_2 \cap \neg A_1)
\ee
By hypothesis $H(\epsilon) <\alpha$, so by \cref{lem:almost_linear}, we can upper bound:
\be 
\mathbb{P}(A_2 \cap A_1) \leq\mathbb{P}(A_1) =O\left(\frac{1}{n^{w\alpha-2}}\right)
\ee

Recall from the previous lemma that the code $C$ is defined through the parity check matrix $H$ as the set of all subsets of the columns of $H$ which sum to zero.  A `bad' event in the current context is the existence of a set of columns which simultaneously sum to zero and has small weight outside the set $K$.  Such an event implies the existence of a codeword which is ``nearly covered up'' by the erasure.  We proceed by bounding the probability that such a set exists.  

Let $s \subset [n]$ be some subset of the columns containing fewer than $ \epsilon' n$ many columns outside the erased set $K$, and let $Q$ be the class of all sets with this property.  Let us define the event $B_s$ as `the sum of the columns in the set $s$ is zero' (or equivalently that the membership vector of the set $s$ forms a word in the code).  It is easy to see that:
\be 
A_2 \Rightarrow \bigcup_{s \in Q} B_s
\ee
so we have immediately that:
\be 
A_2 \cap \neg  A_1 \Rightarrow \left(\bigcup_{s \in Q} B_s \right)\cap \neg A_1
\ee
which implies the upper bound:
\be 
\mathbb{P}( A_2 \cap \neg  A_1) \leq \mathbb{P}\left(\left(\bigcup_{s \in Q} B_s \right)\cap \neg A_1\right)
\ee
Now we will compute an upper bound on $\mathbb{P}(\left(\cup_{s \in Q} B_s \right)\cap \neg A_1)$.  Let $s$ be some subset with $l_1$ many elements in $K$ and $l_2$ many elements in $V\setminus K$.

By \cref{cor:rand_sum_vectors}, the probability is the same as before:
\be 
\mathbb{P}(B_s)=\left(\frac{1+\left(1-\frac{2 w \log(n)}{n} \right)^{l_1+l_2}}{2} \right)^{\alpha n}
\ee
except that under our assumptions (namely that we are in the case $\neg A_1$) we have, 
\be 
1 \leq l_2 \leq \epsilon ' n
\ee
and 
\be 
(\epsilon-\epsilon') n \leq l_1
\ee
so we have:
\be 
\mathbb{P}\left(B_s\cap \neg A_1\right) \leq \frac{1}{2^{\alpha n}}\left(1+\left(1-\frac{2 w\log(n)}{n}  \right)^{(\epsilon-\epsilon') n}  \right)^{\alpha  n} \leq  \frac{1}{2^{\alpha n}}\left(1+\frac{1}{n^{2 w (\epsilon-\epsilon') }} \right)^{\alpha n}
\ee
Where we assumed $n$ was very large for the final inequality.  Using the union bound, we can then argue:
\begin{align}\label{eq:68}
\mathbb{P}\left(\left(\bigcup_{s \in Q}B_s\right)\cap \neg A_1 \right)\leq \binom{(1-\beta)n}{\epsilon' n} \epsilon' n \sum_{k=(\epsilon-\epsilon')n}^{\beta n}\binom{\beta n }{k} \frac{1}{2^{\alpha n}}\left(1+\frac{1}{n^{2 w (\epsilon-\epsilon')}} \right)^{\alpha n}\\
\nonumber \leq \frac{2^{o(n) +H\left(\frac{\epsilon'}{1-\beta}\right)(1-\beta)n}}{2^{(\alpha-\beta) n}}
\end{align}

In the last step we used the standard approximation to binomial coefficients:
\be 
\binom{m}{\delta m}=2^{o(m)}2^{H(\delta)m}
\ee
We obtain an  upper bound that is exponentially small with $n$ if 
\be 
H\left(\frac{\epsilon'}{1-\beta}\right)(1-\beta) < \alpha-\beta 
\ee

\end{proof}

\end{document}